\documentclass[
final
]{dmtcs-episciences}

\usepackage[utf8]{inputenc}
\usepackage{subfigure}

\usepackage{amsmath}
\usepackage{amsthm}
\usepackage{amssymb}
\usepackage{stmaryrd}
\usepackage{subfigure}
\usepackage{tikz}
\usetikzlibrary{calc}

\theoremstyle{plain}
\newtheorem{theorem}{Theorem}
\newtheorem{lemma}[theorem]{Lemma}

\newtheorem{proposition}[theorem]{Proposition}
\theoremstyle{definition}
\theoremstyle{remark}
\newtheorem*{remark}{Remark}

\newcommand{\Mp}{\ensuremath{\mathrm{mp}}}
\newcommand{\Bc}{\ensuremath{\gamma_b}}

\newcommand{\R}{\text{rad}}

\tikzstyle{vertex}=[circle, fill=black, inner sep= 0, minimum size = 4]
\tikzstyle{unselected}=[circle, draw, fill=white, inner sep= 0, minimum size = 4]
\tikzstyle{unknown}=[circle, fill=black, inner sep= 0, minimum size = 2]

\author{Laurent Beaudou\affiliationmark{1}\thanks{Supported by ANR DISTANCIA (ANR-17-CE40-0015) and GraphEn (ANR-15-CE40-0009).}
  \and Richard Brewster \affiliationmark{2}}
\title{On the multipacking number of grid graphs}
\affiliation{
  Higher School of Economics, Moscow, Russian Federation\\
  Thompson Rivers University, Kamloops, BC, Canada.}
\received{2018-4-20}
\revised{2019-3-7}
\accepted{2019-6-7}

\begin{document}
\publicationdetails{21}{2019}{3}{23}{4452}
\maketitle

\begin{abstract}

  In 2001, Erwin introduced \emph{broadcast domination} in graphs. It
  is a variant of classical domination where selected vertices may
  have different domination powers. The minimum cost of a dominating
  broadcast in a graph $G$ is denoted $\Bc(G)$. The dual of this
  problem is called \emph{multipacking}: a multipacking
  is a set $M \subseteq V(G)$ such that for any vertex $v$ and
  any positive integer $r$, the ball of radius $r$ around $v$ contains
  at most $r$ vertices of $M$. The maximum size of a
  multipacking in a graph $G$ is denoted $\Mp(G)$. Naturally $\Mp(G)
  \leq \Bc(G)$. Earlier results by Farber and by Lubiw show that
  broadcast and multipacking numbers are equal for strongly chordal
  graphs.
  
  In this paper, we show that all large grids (height at least $4$ and
  width at least $7$), which are far from being chordal, have their
  broadcast and multipacking numbers equal.
  \end{abstract}


\section*{Introduction}

Given a graph $G$ with vertex set $V$ and edge set $E$, a {\em
  dominating broadcast} of $G$ is a function $f$ from $V$ to
$\mathbb{N}$ such that for any vertex $u$ in $V$, there is a vertex
$v$ in $V$ with $f(v)$ positive and greater than the distance from $u$
to $v$. Define the \emph{ball of radius $r$ around $v$} by $N_r(v) =
\{ u : d(u,v) \leq r \}$.  Thus a dominating broadcast is a cover of
the graph with balls of several positive radii. The \emph{cost} of a
dominating broadcast $f$ is $\sum_{v \in V} f(v)$ and the minimum cost
of a dominating broadcast in $G$, its \emph{broadcast number}, is
denoted $\Bc(G)$.

\begin{remark}
  One may consider the cost to be any function of the powers
  (for example the sum of the squares), see e.g.~\cite{HeggernesLokshtanov2006}. We shall stick to the
  classical convention of linear cost.
\end{remark}

The dual problem of broadcast domination is \emph{multipacking}. A
multipacking in a graph $G$ is a subset $M$ of its vertices such that
for any positive integer $r$ and any vertex $v$ in $V$, the ball of
radius $r$ centred at $v$ contains at most $r$ vertices of $M$.  The
maximum size of a multipacking of $G$, its \emph{multipacking number}, is denoted $\Mp(G)$. 
We may write $\Bc$
and $\Mp$ when the graph in question is clear from context or
unimportant.

Broadcast domination was introduced by Erwin~\cite{Erwin2001,
  Erwin2004} in his doctoral thesis in 2001. Multipacking was then
defined in Teshima's Master's Thesis~\cite{Teshima2012} in 2012, see
also~\cite{Brewster2013}. However, this work fits into the general
study of coverings and packings, which has a rich history in Graph
Theory, see for example the monograph by Cornu\'ejols~\cite{Cornuejols2001}. 

Since minimum dominating broadcast and multipacking are dual problems,
we know that for any graph $G$,
\begin{equation*}
  \Mp(G) \leq \Bc(G).
\end{equation*}
A natural question comes to mind. Under which conditions are they
equal? For example, it is known that strongly chordal graphs have
their broadcast and multipacking numbers equal. This follows from a
primal-dual algorithm of Farber~\cite{Farber84} applied to
$\Gamma$-free matrices, used to solve the (weighted) dominating set
problem for strongly chordal graphs. The work of
Lubiw~\cite{Lubiw82,Lubiw87} shows the vertex-neighbourhood ball
incidence matrix is $\Gamma$-free for strongly chordal graphs, and
hence the primal-dual algorithm can also be used to solve the
broadcast domination problem for strongly chordal graphs.  For trees,
direct proofs of $\Mp(T) = \Bc(T)$ and linear-time algorithms to find
$\Mp(T)$ appear in~\cite{Brewster2013,Brewster2017} (see
also~\cite{Dabney2007,Dabney2009}). For strongly chordal graphs,
Farber's algorithm runs in $O(n^3)$ time. The general broadcast
domination problem can be solved in $O(n^6)$
time~\cite{HeggernesLokshtanov2006}.  In this paper we study grid
graphs which are far from being strongly chordal (or even chordal). We
show the following theorem.
\begin{theorem}
  \label{thm:grids}
  For any pair of integers~$n \geq 4$ and~$m \geq 4$,
  \begin{equation*}
    \Mp(P_n \Box P_m) = \Bc(P_n \Box P_m).
  \end{equation*}
  with the exception of $P_4 \Box P_6$ where $\Mp(P_4 \Box P_6) = 4$
  and $\Bc(P_4 \Box P_6) = 5$.
\end{theorem}
This gives an infinite family of non-chordal graphs for which $\Mp =
\Bc$. Another such family is the cycles of length $0$ modulo $3$.  It
is trivial to verify that $\Mp(C_{3k}) = \Bc(C_{3k}) = k$.

Dunbar et al.~\cite{dun_al_2006} gave the exact value of the broadcast
number for grids.

\begin{theorem}[Dunbar et al.~{\cite[Th. 28]{dun_al_2006}}]
  For any pair of positive integers~$n$ and~$m$, 
  \begin{equation*}
    \Bc(P_n \Box P_m) = \left\lfloor \frac{n}{2} \right\rfloor + \left\lfloor \frac{m}{2} \right\rfloor.
  \end{equation*}
  \label{thm:bcgrids}
\end{theorem}

\begin{remark} 
The value of $\Bc(P_n \Box P_m)$ given by Theorem~\ref{thm:bcgrids} is
the radius of the grid. Since there is always a dominating broadcast
with cost $\R(G)$~\cite{dun_al_2006,Erwin2004}, and our proof of
Theorem~\ref{thm:grids} yields a multipacking of size $\R(G)$, this
paper gives an alternative proof of Theorem~\ref{thm:bcgrids}.
\end{remark}

\section{Preliminaries and small grids}

We use standard notation throughout the paper.  Specific to our work
is the following: the grid $P_n \Box P_m$ has $n$ rows and $m$
columns.  We may also say the grid has height $n$ and length $m$.  The
vertex in row $i$ and column $j$ is denoted $v_{i,j}$.  As a
convention, the vertex $v_{0,0}$ is the bottom, left corner of the
grid.  The integers between $k$ and $\ell$ inclusive are denoted
$\llbracket k,\ell \rrbracket$.

The proof of Theorem~\ref{thm:grids} is technical. In order to
ease the process, we start with an easy counting lemma.

\begin{lemma}
  \label{lem:path}
  Let $G$ be a graph, $k$ be a positive integer and
  $u_0,\ldots,u_{3k}$ be an isometric path in $G$. 
  Let \mbox{$P=\{u_{3i} : i \in \llbracket 0,k \rrbracket \}$} be the
  set of every third vertex on this path. Then, for any positive integer $r$ and any ball
  $B$ of radius $r$ in~$G$,
  \begin{equation*}
    |B \cap P| \leq \left\lceil \frac{2r+1}{3} \right\rceil.
  \end{equation*}
\end{lemma}

\begin{proof}
  Let $B$ be a ball of radius $r$ in $G$, then any two vertices in $B$
  are at distance at most $2r$. Since the path $(u_0,\ldots,u_{3k})$
  is isometric the intersection of the path and $B$ is included in a
  subpath of length $2r$. This subpath contains at most $2r+1$
  vertices and only one third of those vertices can be in $P$.
\end{proof}

For the sake of completeness, we also determine the multipacking
numbers of grids with height 2 and 3.

\begin{proposition}
Let $n$ be a positive integer.  Then
$$
\Mp(P_n \Box P_2) = 
\left\lceil \frac{2n}{5} \right\rceil
$$
\end{proposition}

\begin{proof}
%
Let $P$ be a maximum multipacking of $P_n \Box P_2$. 
We claim that no five consecutive columns contain three members of $P$.
Suppose to the contrary that columns $i$ to $i+4$ contain three members of $P$.
No two consecutive columns each contain a member of $P$, as 
any pair of vertices in $P_2 \Box P_2$ are at distance at most~$2$ apart
(and thus in a ball of radius $1$).
Hence, the three elements are without loss of generality
$\{ v_{i,0}, v_{i+2,1}, v_{i+4,0} \}$.  
However, this implies $|N_2[v_{i+2,0}] \cap P| = 3$, a contradiction.
  
Writing $n=5q+r, 0 \leq r \leq 4$, we conclude that
the first $5q$ columns of the grid contain at most $2q$ elements of $P$.
Next, it is easy to verify that $\Mp(P_1 \Box P_2) = \Mp(P_2 \Box P_2) = 1$,
and $\Mp(P_3 \Box P_2) = \Mp(P_4 \Box P_2) =2$.  
Let $s$ be the number of elements of $P$ in the final $r$ columns
of the grid.  Then, $s = 0$ if $r=0$,
$s \leq 1$ if $r = 1, 2$ and $s \leq 2$ if $r=3,4$. 
Thus, $|P| \leq 2q + \lceil 2r / 5 \rceil$.
Equivalently, $|P| \leq \left\lceil \frac{2n}{5} \right\rceil$.  

On the other hand, consider the set $P$ defined as follows.
\begin{eqnarray*}
v_{i,0} \in P & \mbox{ for } & i \equiv 0 \pmod{5} \\
v_{i,1} \in P & \mbox{ for } & i \equiv 2 \pmod{5} 
\end{eqnarray*}
Consider a ball $B$ of radius $r \geq 2$.  It contains 
vertices from at most $2r+1$ consecutive columns of $P_n \Box P_2$.  
By construction, every five consecutive columns contain at most $2$ 
elements of $P$.  
$$
|B \cap P| \leq 2 \left\lceil \frac{(2r+1)}{5} \right\rceil 
$$
It is straightforward to check that, 
$2 \lceil (2r+1)/ 5 \rceil \leq r$ for $r \neq 1, 3, 5$. 
(For $r<10$, simply evaluate  $2 \lceil (2r+1)/ 5 \rceil$.
For $r \geq 10$, $2 \lceil (2r+1)/ 5 \rceil \leq 2 (2r/5 + 1) \leq r$.)
It is easy to check that each ball of radius~$1$ contains at most one
element of $P$. When $r=3$, $B$ contains vertices from $7$ columns of 
which at most $3$ columns may contain packing vertices.  Similarly, 
when $r=5$, we observe that any $11$ consecutive columns contain at 
most $5$ packing vertices.
\end{proof}

We now turn to the special case when $m=3$. The following result gives
$\Mp(P_n \Box P_3)$.  Since $\Bc(P_n \Box P_3) = \lfloor \frac{n}{2}
\rfloor + 1$, we note that $\Mp = \Bc$ for $n \not\equiv 0 \pmod{4}$.

\begin{proposition}
Let $n$ be a positive integer.  Then
$$ \Mp(P_n \Box P_3) = \left\{ \begin{array}{ll} \left\lfloor
  \frac{n}{2} \right\rfloor & \mbox{ if } n \equiv 0 \pmod{4} \\[4pt]
  \left\lfloor \frac{n}{2} \right\rfloor + 1 & \mbox{ if } n \equiv
  1,2,3 \pmod{4}
\end{array} \right .
$$
\end{proposition}

\begin{proof}
Since $\mathrm{rad}(P_n \Box P_3) = \lfloor n/2 \rfloor + \lfloor 3/2
\rfloor$, we know that $\Mp(P_n \Box P_3) \leq \lfloor n/2 \rfloor + 1$.
Given a multipacking $P$ of $P_n \Box P_3$ and any four consecutive
columns say $i, i+1, i+2, i+3$, if $P$ has three members in
these columns, then without loss of generality they belong to columns $i,
i+2, i+3$.  Moreover, we can assume that the two in columns $i+2, i+3$ are
$v_{i+2,2}$ and $v_{i+3,0}$.  The only vertices that are not within
distance~$2$ of either of these two packing vertices are $v_{i,0}$ and $v_{i,1}$. 
However, all three of these vertices are in a ball of radius~$2$ centred at
$v_{i+2,0}$ in the former case and $v_{i+2,1}$ in the latter, a contradiction. 
Thus, the four columns contain at most $2$ packing vertices. 
Specifically, in the case $n = 4q$, $\Mp(P_n \Box P_3) \leq 2q = \lfloor n/2 \rfloor$.

On the other hand, consider the set $P$ defined as follows.
\begin{eqnarray*}
v_{i,0} \in P & \mbox{ for } & i \equiv 0 \pmod{4} \\
v_{i,2} \in P & \mbox{ for } & i \equiv 1 \pmod{4} 
\end{eqnarray*}
As the minimum distance between vertices in $P$ is $3$, no ball of
radius~$1$ contains more than one element of $P$.  Consider a ball $B$
of radius $r \geq 2$.  The ball contains vertices from at most $2r+1$
consecutive columns.  We need to confirm that the ball has at most $r$ elements of $P$.
First, suppose that $r = 2t$.  By symmetry, we may assume that the left most column
of $B$ is in $\{ 0, 1, 2, 3 \}$.  If the left most column is $2$ or $3$,
then $B$ contains vertices from columns $\{ 4, 5, \dots, 4t + 2, 4t + 3 \}$.
Each contiguous block of four columns contains two members of $P$, giving $B$ 
has a total of at most $2t = r$ vertices of $P$.  
If the left most column of $B$ is $0$ or $1$,
then $B$ covers columns $0, 1, \dots, 4t$ or $1, \dots, 4t, 4t+1$.
In both cases, $B$ has exactly $2t+1$ columns with a vertex of $P$.
However, in both cases $v_{1,2}$ and $v_{4t,0}$ are at distance $4t+1 = 2r+1$
apart and thus, at most one belongs to $B$.  In all cases, $|B \cap P| \leq r$.
If $r = 2t+1$, the analysis is similar.  
Either the $4t+3$ columns of $B$
contain at most $2t+1 = r$ vertices of $P$, or 
the ball $B$ has $2t+2 = r+1$ columns containing vertices of $P$, but there is a pair
(for example $\{ v_{0,0}, v_{4t+1,2} \}$) at distance $2r+1$, in which case
$B$ itself contains at most $r$ vertices of $P$.
\end{proof}

\section{Multipacking number for large grids}
\label{sec:grid}

In this section, we prove Theorem~\ref{thm:grids}. The radius of a
grid graph $P_n \Box P_m$ is $\lfloor \frac{n}{2} \rfloor + \lfloor
\frac{m}{2} \rfloor$. Since the broadcast number of a graph is at most
its radius, it is sufficient to find a multipacking of size $\lfloor
\frac{n}{2} \rfloor + \lfloor \frac{m}{2} \rfloor$. We now proceed with
the construction of such multipackings.

\subsection{Restriction to even sizes}

First, we shall prove that we can restrict ourselves to cases when $n$
and $m$ are both even numbers. Because of the singularity for the grid
of size $4 \times 6$, we need to check the grids of sizes $5 \times 6$
and $4 \times 7$ by hand (see Figure~\ref{fig:5x6}). Now, suppose that
$n$ is odd. Then, $n-1$ is even and is at least~$4$. Moreover the $n-1
\times m$ grid is not $4 \times 6$ since we ruled out the $5 \times 6$
and $4 \times 7$ cases. Thus, if we know that the grid of size $n-1
\times m$ has a multipacking of size $\frac{n-1}{2} + \lfloor
\frac{m}{2} \rfloor$, which is equal to $\lfloor \frac{n}{2} \rfloor +
\lfloor \frac{m}{2} \rfloor$, we can add an empty column in the middle
of this grid. We obtain a multipacking of the desired size for our
grid. Indeed, given a vertex $v$ from the smaller grid, the ball of
radius $r$ with centre $v$ in the larger grid only contains vertices
of the packing which were at distance at most $r$ from $v$ in the
smaller grid. A ball of radius $r$ centred at a vertex of the new
column only contains vertices of the packing which are within distance
$r$ of both its neighbours from the former grid. Thus, these balls
cannot contain more than $r$ elements of the packing which satisfies
our claim. The same reasoning works for $m$.

  \begin{figure}[ht]
    \begin{center}
      \begin{tikzpicture}[scale=0.5]
        \foreach \i in {0, ..., 5}
        \foreach \j in {0,1,2,3,4}
        \node[inner sep=0] (u\i\j) at (\i,\j) {};
        
        \foreach \i in {0, ..., 5}
        \draw (u\i0)--(u\i4);
        
        \foreach \j in {0, ..., 4}
        \draw (u0\j)--(u5\j);
        
        \foreach \i in {0,1,2,3,4,5}
        \foreach \j in {0,1,2,3,4}
        \node[unselected] at (u\i\j) {};
        
        \node[vertex] at (u00) {};      
        \node[vertex] at (u04) {};
        \node[vertex] at (u50) {};
        \node[vertex] at (u54) {};  
        \node[vertex] at (u23) {};    
        
        \begin{scope}[xshift=8cm]
          \foreach \i in {0,..., 6}
          \foreach \j in {0,...,3}
          \node[inner sep=0] (u\i\j) at (\i,\j) {};
          
          \foreach \i in {0, ..., 6}
          \draw (u\i0)--(u\i3);
          
          \foreach \j in {0,..., 3}
          \draw (u0\j)--(u6\j);
          
          \foreach \i in {0,..., 6}
          \foreach \j in {0,...,3}
          \node[unselected] at (u\i\j) {};
          
          \node[vertex] at (u00) {};      
          \node[vertex] at (u60) {};
          \node[vertex] at (u63) {};
          \node[vertex] at (u03) {};  
          \node[vertex] at (u31) {};
        \end{scope}  
      \end{tikzpicture}
    \end{center}
    \caption{Multipackings of order 5 for grids of size $5 \times 6$
      and $4 \times 7$}
    \label{fig:5x6}
  \end{figure}
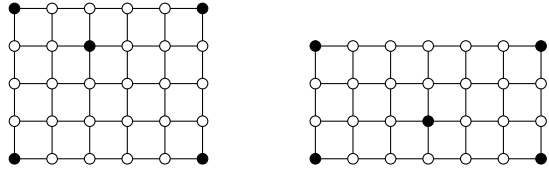

  The remainder of the proof is concerned with grids with even
  dimensions. Small cases require some specific care so that we will
  treat them after the general case. In all cases, we shall use a
  systematic way of selecting vertices along the sides of the grid. We
  describe them in the following paragraph.

  \subsection {The $i$-pattern}
  
  Fix an integer $i$. Given a path $v_0 v_1 \ldots v_{z-1}$ of order
  $z$ greater than or equal to $3i$, the {\em $i$-pattern} on this
  path consists in selecting every third vertex from $v_0$ to
  $v_{3(i-1)}$ and then every fourth vertex starting from $v_{3i}$ (if
  it exists). Note that the $i$-pattern on a path of order $z$ selects
  exactly $i$ vertices from the beginning and one fourth (rounded
  up) of the rest. This amounts to $i + \lceil \frac{z-3i}{4}
  \rceil$ which can be simplified.
  \begin{equation}
    \text{The }i\text{-pattern on a path of order }z \text{ selects exactly } \left\lceil \frac{z+i}{4} \right\rceil \text{ vertices.}
    \label{eq:size_pattern}
  \end{equation}
  Moreover, the density of the $i$-pattern is bounded above by a
  function of $i$. By this, we mean that a subpath of length $\ell$ of
  $v_0 v_1 \ldots v_{z-1}$ cannot hit too many vertices of the
  $i$-pattern. If $\ell$ is at least $3i$, it could take the whole
  beginning ($i$ vertices) and a fourth of the rest. This amounts to
  $i + \lceil \frac{\ell + 1 - 3i}{4} \rceil$ which equals $\lceil
  \frac{\ell+1+i}{4} \rceil$. Whenever $\ell$ is strictly less than
  $3i$, it would take at most $\lceil\frac{\ell+1}{3}\rceil$
  vertices. But in that case, 
  \begin{flalign*}
    && \left\lceil\frac{\ell+1}{3}\right\rceil & \leq \left\lceil\frac{4\ell+4}{12}\right\rceil&&\\
    && & \leq \left\lceil\frac{3\ell+3+\ell+1}{12}\right\rceil&&\\
    && & \leq \left\lceil\frac{3\ell+3+3i}{12}\right\rceil&\text{(since } \ell + 1 \leq 3i \text{)}&\\
    && & \leq \left\lceil\frac{\ell+1+i}{4}\right\rceil.&&
  \end{flalign*}
  In the end, we may state that
  \begin{equation}
    \text{a subpath of length } \ell \text{ hits at most } \left\lceil\frac{\ell+1+i}{4}\right\rceil \text{ vertices on a }i\text{-pattern.}
    \label{eq:density_pattern}
  \end{equation}

  \subsection{Large grids}
  \label{sec:large}
  As said before, small grids require some extra-care. In this part,
  we only consider grids with dimensions at least $8$ in both directions. Fix
  $n$ and $m$ two even integers greater than or equal to $8$. We let
  $k=n/2$ and $k'=m/2$. We view each side of the grid
  as a path from which we remove the last three vertices (see
  Figure~\ref{fig:grid}). In these paths, we pack an adequate
  number of vertices using a specific $i$-pattern. Finally,
  we will estimate an upper bound on the number of such vertices in a
  ball of size $r$. This will cover most of the radii but the
  last few ones will be treated using some tailor-made arguments.
  \begin{figure}[ht]
    \scriptsize
    \begin{center}
      \begin{tikzpicture}
        \node[vertex] (x) at (-3,2) {};
        \node[unselected] (x2) at (2.6,2) {};      
        \node[unselected] (x1) at (2.2,2) {};
        \node[vertex] (z) at (3,-2) {};
        \node[unselected] (z2) at (-2.6,-2) {};
        \node[unselected] (z1) at (-2.2,-2) {};
        \node[vertex] (y) at (3,2) {};
        \node[unselected] (y2) at (3,-1.6) {};
        \node[right] at (y2) {$(n-1,2)$};
        \node[unselected] (y1) at (3,-1.2) {};
        \node[right] at (y1) {$(n-1,1)$};
        \node[vertex] (t) at (-3,-2) {};
        \node[right] at (z) {$(n-1,0)$};
        \node[unselected] (t2) at (-3,1.6) {};
        \node[unselected] (t1) at (-3,1.2) {};
        \node[below left] at (t) {$(0,0)$};
        \node[above right] at (y) {$(n-1,m-1)$};
        \draw (x) -- (x1) -- (x2) -- (y);
        \draw (y) -- (y1) -- (y2) -- (z);
        \draw (z) -- (z1) -- (z2) -- (t);
        \draw (t) -- (t1) -- (t2) -- (x);
        \draw (-3.1,2.1) rectangle (2,1.9);
        \draw (3.1,-2.1) rectangle (-2,-1.9);
        \draw (3.1,2.1) rectangle (2.9,-1);      
        \draw (-3.1,-2.1) rectangle (-2.9,1);      
      \end{tikzpicture}
    \end{center}
    \caption{General sketch, packing on the perimeter.}
    \label{fig:grid}
  \end{figure}
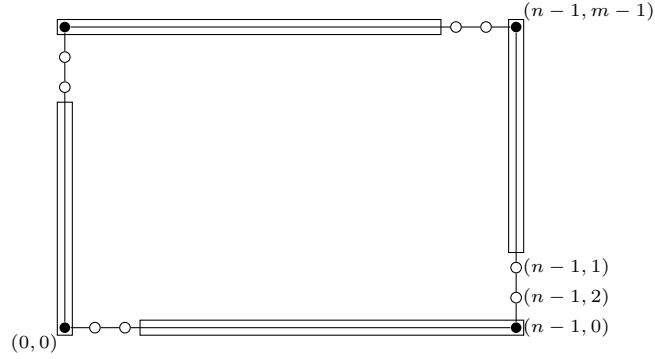
  
  We use the $i$-patterns (where $i = 0, 1,$ or $2$) to select vertices on 
  ``horizontal'' and  ``vertical'' sides.   
  The packing on the horizontal sides is as follows (with the vertical sides
  being similar).  Our choice of $i$ depends on the parity of $k$.  
  In all cases, $2k-4 \geq 3i$.
  If $k$ is even, then we use a $1$-pattern on the top 
  $(0,m-1), (3,m-1), (7,m-1), \dots, (n-5,m-1)$ and a $1$-pattern
  on the bottom $(n-1,0), (n-4,0), (n-8,0), \dots, (4,0)$.  
  In this case we shall write $i_t = i_b = 1$.  If $k$ is odd,
  we use a $2$-pattern on the top $(0,m-1), (3,m-1), (6,m-1), (10,m-1), \dots, (n-4,m-1)$
  and a $0$-pattern on the bottom $(n-1,0), (n-5,0), \dots, (5,0)$.
  In this case we write $i_t = 2$ and $i_b = 0$.
  Using~\eqref{eq:size_pattern}, we see there are exactly $\frac{k}{2}$
  vertices selected on a (horizontal) side, when $k$ is even.
  When $k$ is odd $\lfloor \frac{k}{2} \rfloor$ are selected on the bottom while 
  $\lceil \frac{k}{2} \rceil$ are selected on the top.  In all cases $\frac{n}{2}=k$
  vertices are selected. These selections are depicted on Figure~\ref{fig:selpath}
  for $n=16$ and $n=18$ (only the top and bottom sides of the grid are
  drawn). We call $H$ the set of vertices selected on the horizontal
  paths. Similarly we select a total of $k'$ vertices on the vertical 
  sides and let $V$ denote these vertices. 
  After this process, we have a set $P$ of $k+k'$
  vertices. We shall prove that it is a valid multipacking.
  \begin{figure}[ht]
    \scriptsize
    \begin{center}
      \begin{tikzpicture}[xscale=.5]
        \foreach \s in {0,2}{
          \pgfmathsetmacro{\sh}{\s*.8};
          \pgfmathtruncatemacro{\sn}{16+\s}
          \node at (-.5,0-\sh) {$n=\sn$};
          \draw (1,.3-\sh) -- (16+\s,.3-\sh);
          \draw (1,-.3-\sh) -- (16+\s,-.3-\sh);
          \draw (.8,.4-\sh) rectangle (13.2+\s,.2-\sh);
          \draw (3.8,-.4-\sh) rectangle (16.2+\s,-.2-\sh);
          \foreach \i in {1,...,\sn}{
            \node[unselected] (x\s\i) at (\i,.3-\sh) {};
            \node[unselected] (y\s\i) at (17+\s-\i,-.3-\sh) {};
          }
        }
        
        \foreach \i in {1,4,8,12}{
          \node[vertex] at (x0\i) {};
          \node[vertex] at (y0\i) {};
        }
        
        \foreach \i in {1,4,7,11,15}{
          \node[vertex] at (x2\i) {};
        }
        \foreach \i in {1,5,9,13}{
          \node[vertex] at (y2\i) {};
        }
        
      \end{tikzpicture}
    \end{center}
    \caption{Selection of vertices on horizontal paths}
    \label{fig:selpath}
  \end{figure}
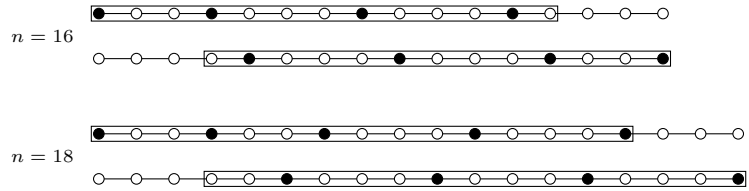
  
  \paragraph{Most balls are valid.}
  Let $r$ be some integer between 1 and $k+k'-5$, and let~$B$ be a ball
  of radius~$r$. 
  \begin{itemize}
  \item  If $B$ does not intersect any side of the grid, then
    its intersection with~$P$ is empty and $|B \cap P| \leq r$
    trivially. 
    
  \item Suppose now that $B$ intersects only one side of the
    grid, or two consecutive sides. Then its intersection with $P$ lies
    on an isometric path of the grid where selected vertices are at
    distance at least 3 from each other. Thus the cardinality of $B \cap
    P$ is bounded above by $\left\lceil\frac{2r+1}{3}\right\rceil$ which
    is at most $r$ since $r$ is a positive integer (see
    Lemma~\ref{lem:path}).
    
  \item Now if $B$ intersects two opposite sides of the grid (let them
    be top and bottom), let $y$ denote the ordinate of the center of
    $B$. Recall that bottom has ordinate 0 while top has ordinate
    $m-1$. Now observe that the metric induced by the grid is similar
    to $\ell_1$ metric. Thus $B$ intersects the bottom side on a
    subpath of length at most $2(r-y)$ and the top side on a subpath
    of length at most $2(r-2k'+1+y)$. We claim that in most cases $|B
    \cap H| \leq r-k'+2$. Only in very specific cases can $|B \cap H|$
    be equal to $r-k'+3$. 

    Let $x_b$ be the difference between $2(r-y)$ and the actual length
    of the intersection between $B$ and the bottom part of $H$ (recall
    that the bottom part of $H$ does not include the last three
    vertices as depicted on Figures~\ref{fig:grid}
    and~\ref{fig:selpath}). Similarly, we define $x_t$ for the top
    part of $H$. Then $B$ intersect the bottom part of $H$ on a
    subpath of length $2(r-y) - x_b$. We now use
    \eqref{eq:density_pattern}:
    \begin{align*}
      |B \cap H| & \leq \left\lceil\frac{2(r-y) - x_b + 1 + i_b}{4} \right\rceil + \left\lceil\frac{2(r-2k'+1+y) -x_t + 1 + i_t}{4} \right\rceil \\
      & \leq \left\lceil\frac{2(r-y) + 1 + i_b - x_b}{4} \right\rceil + \left\lceil\frac{2(r+y)+3+i_t - x_t}{4} \right\rceil - k'.
    \end{align*}
    Notice that $r-y$ and $r+y$ have same parity. First suppose they
    are both odd. Then $2(r-y) - 2$ and $2(r+y) - 2$ are multiples of
    4. We can rewrite our bound.
    \begin{align*}
      |B \cap H| & \leq \left\lceil\frac{2(r-y) - 2 + 3 + i_b - x_b}{4} \right\rceil + \left\lceil\frac{2(r+y)-2+5+i_t - x_t}{4} \right\rceil - k' \\
      & \leq \left\lceil\frac{3 + i_b - x_b}{4} \right\rceil + \left\lceil\frac{5+i_t - x_t}{4} \right\rceil +r - 1 - k'.
    \end{align*}
    Our pattern choice is either $i_b = i_t = 1$ or $i_b = 0$ and $i_t
    = 2$. In both cases, the ceilings add up to at most 3. So $|B \cap
    H| \leq r - k' + 2$ when $r-y$ is odd. Now suppose that $r-y$ is
    even. Then the rewriting is straightforward.
    \begin{align*}
      |B \cap H| & \leq \left\lceil\frac{2(r-y) + 1 + i_b - x_b}{4} \right\rceil + \left\lceil\frac{2(r+y)+3+i_t - x_t}{4} \right\rceil - k' \\
      & \leq \left\lceil\frac{1 + i_b - x_b}{4} \right\rceil + \left\lceil\frac{3+i_t - x_t}{4} \right\rceil +r - k'.
    \end{align*}
    When the pattern is $i_b = i_t = 1$, ceilings add up to at most 2
    and once again $|B \cap H| \leq r - k' + 2$. When $i_b = 0$ and
    $i_t = 2$, ceilings can unfortunately sum up to 3. But for this,
    both $x_b$ and $x_t$ must be 0. Let us be more precise. It means
    that both bottom and top intersections must be full (subpaths of
    length $2(r-y)$ and $2(r-2k'+1+y)$). Moreover, $B \cap H$ must use the
    corner vertex in the top part since otherwise, it would be
    intersecting a path which is nothing but a 1-pattern (or a
    0-pattern). As a consequence, the center of $B$ must be at
    distance exactly $r$ from the top left corner. Moreover, $B$
    cannot reach the top right corner of the grid (otherwise, $x_t$
    would be strictly positive). Similarly, $B$ cannot reach the
    bottom left corner since it is out of the bottom part of $H$ and
    it would require $x_b$ to be strictly positive. 

    \begin{itemize}
    \item If $B$ intersects only the top and bottom part of $H$, then
      $|B\cap H| \leq r-k'+3 \leq r$ since $k'$ is at least 4.

    \item If $B$ intersects also exactly one vertical side. This one
      can contribute at most for $\lceil \frac{k'}{2} \rceil$ (by our
      choice of $P$). Thus, in most cases
      \begin{align*}
        |B \cap P| & \leq r - k' + 2 + \left\lceil \frac{k'}{2} \right\rceil\\
        & \leq r - k' + 2 + \frac{k'}{2} + \frac{1}{2}\\
        & \leq r - \frac{1}{2} (k'-5) 
      \end{align*}
      which is at most $r$ since $k'$ is not less than 4 (when $k' =
      4$ we observe $|B \cap P| \leq r + \frac{1}{2}$ implies $|B \cap
      P| \leq r$ since $|B \cap P|$ is an integer).  In the special
      case when $|B \cap H|$ is $r - k' + 3$, recall that the vertical
      side cannot use the corner so it contributes at most for $\lceil
      \frac{k'}{2} \rceil-1$ and the same conclusion holds.

     \item Finally, if $B$ intersects all four sides, we may use the
       corner observation to state that at most one of the directions
       (vertical or horizontal) can contribute for $r-k'+3$ (or
       $r-k+3$). The other direction contributes at most for $r-k+2$
       (or $r-k'+2$) so that
      \begin{align*}
        |B \cap P| &= |B \cap H| + |B \cap V|\\
        &\leq 2r - (k+k') + 5.
      \end{align*}
      This quantity is less than or equal to $r$ whenever $r$ is
      $k+k'-5$ or less.
      
    \end{itemize}

  \end{itemize}

  \paragraph{Balls with a big radius.}
  To finish our proof, we only need to verify that balls with a radius
  $r$ between $k+k'-4$ and $k+k'-1$ verify our constraint.

  Let us treat the maximum radius $k+k'-1$. Note that since $n$ and
  $m$ are both even, this grid, if seen as a chequerboard, has two 
  diagonally opposite white corners and two diagonally opposite black
  corners.  Suppose a ball of radius $k+k'-1$
  contains all the vertices of $P$. Then it must contain the four
  corners of the grid. Since opposite corners are at distance $2k +
  2k' -2$ it means that the centre of the ball is the middle vertex of
  a shortest path between opposite corners. But this middle vertex
  must be white for one pair of corners and must be black for the 
  other pair, which is impossible. Thus every ball of radius $k+k'-1$
  misses at least one corner.

  Now consider a ball of radius $k+k'-2$. Since both pairs of opposite
  corners are at distance $2k + 2k' -2$, at most one corner of each
  pair can be in a ball of such radius. Thus, such a ball misses at
  least two corners.

  Concerning radius, $k+k'-4$, we can match each corner vertex with
  the second selected vertex from the opposite side. The distance
  between them is $2k+2k'-5$ (corner to a 2-pattern or to a 1-pattern)
  or $2k+2k'-6$ (corner to a 0-pattern) depending on the chosen
  pattern. In any case, since vertices in the ball cannot have
  distance more $2k+2k'-8$, such a ball misses at least 4 vertices
  from the total and is valid.

  Finally, we are left with balls of radius $k+k'-3$. We may again
  consider the same matching. If $k$ or $k'$ is even, we have at least
  one direction with two 1-patterns and so at least three of the pairs
  are at distance $2k+2k'-5$. So the ball misses at least three
  vertices and is valid. The last case is when both $k$ and $k'$ are
  odd. In that case, our matching has two pairs at distance $2k+2k'-5$
  (from which the ball misses at least two vertices) and two pairs
  at distance $2k+2k'-6$. As for radius $k+k'-1$, both last pairs are
  on two different colors of the chequerboard (black and white) so
  that at least one of the four concerned vertices is missed. In the
  end, the ball misses at least three vertices and is valid.

  This concludes the proof for grids with sizes at least 8 in both
  directions.

  \subsection{Long grids}
  \label{sec:long}

  The previous discussion leaves out all grids with one of their
  dimensions either 4 or 6. In this section, we provide a way of
  tackling long grids (for which $k \geq 3k'- \ell$ where $\ell$
  depends on the parity of $k+k'$). In the end,
  there will only remain four cases to study.

  We shall pack vertices only on the top and bottom sides of the grid. We
  consider the whole sides (not the $2k-3$ first vertices as in
  Subsection~\ref{sec:large}). 
  Recall to pack an $i$-pattern on a horizontal side requires $3i \leq n-1$.   
  If $k$ and $k'$ have same
  parity, we use a $(2k'-3)$-pattern on both top and bottom sides. 
  This requires $3(2k'-3) \leq n-1$ or $3k'-4 \leq k$.
  If $k$ and $k'$ have different parities, 
  we use a $(2k'-5)$-pattern on one side (say bottom) and
  a $(2k'-1)$-pattern on the other. This requires $3k'-1 \leq k$.  
  By~\eqref{eq:size_pattern}, this process selects
  \begin{equation*}
    \left\lceil\frac{2k+2k'-3}{4}\right\rceil + \left\lceil\frac{2k+2k'-3}{4}\right\rceil \text{ or } \left\lceil\frac{2k+2k'-1}{4}\right\rceil + \left\lceil\frac{2k+2k'-5}{4}\right\rceil
  \end{equation*}
  vertices. In both cases, this can be simplified as $k+k'$ (in the
  first case, $k+k'$ is even, while it is odd in the latter).
  
  Now, if a ball $B$ of radius $r$ intersects only one horizontal side
  of the grid, this intersection lies on an isometric path from which
  we selected at most every third vertex. Then by
  Lemma~\ref{lem:path}, it cannot contain strictly more than $r$
  vertices. Suppose that the ball $B$ intersects both paths. Like in
  the previous subsection, if this ball has its centre on a vertex with
  ordinate $y$ (0 being the bottom and $m-1$ being the top), then it
  intersects the bottom on a path of length at most $2(r-y)$ and the
  top on a path of length at most $2(r-2k'+1+y)$. Then we
  use~\eqref{eq:density_pattern}. If both sides are packed with
  $(2k'-3)$-patterns,
  \begin{align*}
    |B \cap H| & \leq \left\lceil\frac{2(r-y) + 2k'-2}{4} \right\rceil + \left\lceil\frac{2(r-2k'+1+y) + 2k'-2}{4} \right\rceil \\
    & \leq  \left\lceil\frac{2(r-y+k')-2}{4} \right\rceil + \left\lceil\frac{2(r-k'+y)}{4} \right\rceil.
  \end{align*}
  And since $r-y+k'$ and $r+y-k'$ have same parity, one of the
  ceilings adds $\frac{1}{2}$, and
  \begin{align*}
    |B \cap H| & \leq \frac{2(r-y+k')-2}{4} + \frac{2(r-k'+y)}{4} + \frac{1}{2} \\
    & \leq r.
  \end{align*}
  Similarly, if we use the $(2k'-5)$-pattern on bottom and the
  $(2k'-1)$-pattern on the top, we have 
  \begin{align*}
    |B \cap H| & \leq \left\lceil\frac{2(r-y) + 2k'-4}{4} \right\rceil + \left\lceil\frac{2(r-2k'+1+y) + 2k'}{4} \right\rceil \\
    & \leq  \left\lceil\frac{2(r-y+k')-4}{4} \right\rceil + \left\lceil\frac{2(r-k'+y)+2}{4} \right\rceil.
  \end{align*}
  Once again, the rounding adds at most $\frac{1}{2}$ and 
  \begin{equation*}
    |B \cap H|  \leq r.
  \end{equation*}

  When $k' = 2$ and $k$ is even, we use a $2k'-3 = 1$ pattern.
  Thus the previous paragraph is valid for all even $k \geq 2$.
  When $k$ is odd we use a 0-pattern and a 3-pattern. This requires
  $k \geq 5$.  In particular, we have valid multipackings for $4 \times n$
  for any even $n \geq 4$ and $n \neq 6$.  
  In the same manner the previous paragraph gives a valid 
  multipacking of order $k+k'$ when $k' = 3$ provided 
  $k \geq 8$ for even $k$ and $k \geq 5$ for odd $k$. 
  Consequently we have packings of grids with dimensions $6 \times n$ 
  for even $n \neq 6, 8, 12$.  This concludes the proof for long grids.  
  (We remark the above arguments
  show for a fixed $k'$ and sufficiently large $k$, there is an 
  optimal multipacking selecting vertices only on the horizontal sides.)
  
  \subsection{Remaining cases}
  \label{sec:remaining}
  Subsection~\ref{sec:large} covers large grids ($4 \leq k \leq k'$),
  and Subsection~\ref{sec:long} covers long grids ($2 \leq k' \leq
  (k+\ell)/3$). There are four remaining cases that can be checked by
  hand, and have been verified using SageMath.
  Three are depicted on Figure~\ref{fig:remaincases}.
  
  \begin{figure}[ht]
  
  \begin{tikzpicture}[scale=0.4]
  \foreach \i in {0, ..., 5}
    \foreach \j in {0,1,2,3,4,5}
        \node[inner sep=0] (u\i\j) at (\i,\j) {};
        
  \foreach \i in {0, ..., 5}
    \draw (u\i0)--(u\i5);
    
  \foreach \j in {0, ..., 5}
    \draw (u0\j)--(u5\j);
    
  \foreach \i in {0,1,2,3,4,5}
    \foreach \j in {0,1,2,3,4,5}
        \node[unselected] at (u\i\j) {};
        
  \node[vertex] at (u00) {};      
  \node[vertex] at (u05) {};
  \node[vertex] at (u50) {};
  \node[vertex] at (u55) {};  
  \node[vertex] at (u12) {};    
  \node[vertex] at (u42) {};  
  
  \begin{scope}[xshift=8cm]
    \foreach \i in {0, ..., 7}
    \foreach \j in {0,1,2,3,4,5}
        \node[inner sep=0] (u\i\j) at (\i,\j) {};
        
  \foreach \i in {0, ..., 7}
    \draw (u\i0)--(u\i5);
    
  \foreach \j in {0, ..., 5}
    \draw (u0\j)--(u7\j);
    
  \foreach \i in {0, ..., 7}
    \foreach \j in {0,1,2,3,4,5}
        \node[unselected] at (u\i\j) {};
        
  \node[vertex] at (u00) {};      
  \node[vertex] at (u05) {};
  \node[vertex] at (u70) {};
  \node[vertex] at (u75) {};  
  \node[vertex] at (u30) {};    
  \node[vertex] at (u35) {};   
  \node[vertex] at (u63) {};
  \end{scope}  
  
   \begin{scope}[xshift=19cm]
    \foreach \i in {0, ..., 11}
    \foreach \j in {0,1,2,3,4,5}
        \node[inner sep=0] (u\i\j) at (\i,\j) {};
        
  \foreach \i in {0, ..., 11}
    \draw (u\i0)--(u\i5);
    
  \foreach \j in {0, ..., 5}
    \draw (u0\j)--(u11\j);
    
  \foreach \i in {0, ..., 11}
    \foreach \j in {0,1,2,3,4,5}
        \node[unselected] at (u\i\j) {};
        
  \node[vertex] at (u00) {};      
  \node[vertex] at (u05) {};
  \node[vertex] at (u40) {};
  \node[vertex] at (u55) {};  
  \node[vertex] at (u70) {};    
  \node[vertex] at (u85) {};   
  \node[vertex] at (u110) {};
  \node[vertex] at (u115) {}; 
  \node[vertex] at (u23) {};    
  \end{scope}             
  \end{tikzpicture}
  
  \caption{Multipacking for $6 \times 6$, $8 \times 6$, and $12 \times 6$ grids.}
    \label{fig:remaincases}
  \end{figure}
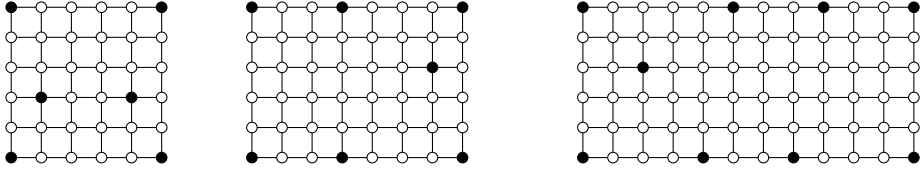

  Finally, the $6 \times 4$ grid is the
  only grid with dimensions at least 4 and multipacking number strictly
  smaller than expected. It is 4 while its broadcast domination number
  is 5. This completes the proof of Theorem~\ref{thm:grids}.

\section*{Acknowledgement}

This research was initiated when the first author and Florent Foucaud
visited the second author at Thompson Rivers University (Kamloops, BC,
Canada) in June 2016. First author is grateful to Universit\'e Clermont
Auvergne which was his institution at that time.

Both authors are grateful to anonymous referees who pointed out a flaw
in the proof of the main theorem in the first submitted version of
this article.



\begin{thebibliography}{00}
\bibitem{Brewster2013} R.~C.~Brewster, C.M.~Mynhardt and L.~Teshima,
  New bounds for the broadcast domination number of a graph,
  \emph{Central European Journal of Mathematics}, {\bf 11} (2013),
  1334--1343.

\bibitem{Brewster2017} R.~C.~Brewster, G.~MacGillivray and F.~Yang, 
Broadcast domination and multipacking in strongly chordal graphs,  
\emph{Discrete Applied Mathematics}, accepted.

\bibitem{Cornuejols2001} G.~Cornu\'{e}jols. {\em Combinatorial
  Optimization: packing and covering.} CBMS-NSF regional conference
  series in applied mathematics, vol. 74. Society for Industrial and
  Applied Mathematics (SIAM), Philadelphia, 2001.

\bibitem{Dabney2007} J.~Dabney 
{\em A linear-time algorithm for broadcast domination in a tree}, 
Master's Thesis, Clemson University, 2007.

\bibitem{Dabney2009}
J.~Dabney, B.C.~Dean, S.T.~Hedetniemi, 
A linear-time algorithm for broadcast domination in a tree. 
\emph{Networks}, \textbf{53} (2009), 160--169.
    
\bibitem{dun_al_2006} J.~E.~Dunbar, D.~J.~Erwin, T.~W.~Haynes,
S.~M.~Hedetniemi and S.~T.~Hedetniemi, 
Broadcasts in graphs, 
\emph{Discrete Applied Mathematics}, \textbf{154} (2006),~59--75.

\bibitem{Erwin2001} D.~J.~Erwin, 
\emph{Cost domination in graphs}, PhD Thesis, Department of Mathematics, Western Michigan University, 2001.

\bibitem{Erwin2004} D.~J.~Erwin, 
Dominating broadcasts in graphs, 
\emph{Bulletin of the ICA}, {\bf 42} (2004), 89--105.

\bibitem{Farber84} M.~Farber, 
Domination, Independent Domination, and Duality in Strongly Chordal Graphs, 
\emph{Discrete Applied Mathematics}, {\bf 7} (1984), 115--130.
   
\bibitem{HeggernesLokshtanov2006} P.~Heggernes and D.~Lokshtanov, Optimal broadcast domination in polynomial time, \emph{Discrete Mathematics},  {\bf 306} (2006), 3267--3280.

\bibitem{Lubiw82} A.~Lubiw, 
{\em $\Gamma$-free Matrices}, 
Master's Thesis, Department of Combinatorics and Optimization, University of Waterloo, 1982.

\bibitem{Lubiw87} A.~Lubiw, 
Doubly Lexical Orderings of Matrices, 
\emph{SIAM Journal on Computing}, {\bf 16} (1987), 854--879.

\bibitem{Teshima2012} L.~Teshima, 
{\em Broadcasts and multipackings in graphs}, 
Master's Thesis, Department of Mathematics and Statistics, University of Victoria, 2012.

\end{thebibliography}

\end{document}